\documentclass[11pt,a4paper]{article}

\usepackage[T1]{fontenc}
\usepackage{lmodern}
\usepackage[utf8]{inputenc}
\usepackage[english]{babel}
\usepackage{amsmath,amssymb,amsthm}
\usepackage[margin=2cm,bottom=2.4cm]{geometry}
\usepackage{array}
\usepackage{microtype}
\usepackage[hidelinks]{hyperref}
\usepackage{fancyhdr}

\newcommand{\sem}[1]{[\![#1]\!]}
\newcommand{\Aapp}{A_{\mathrm{app}}}
\newcommand{\Aex}{A_{\mathrm{ex}}}
\newcommand{\Afin}{A_{\mathrm{fin}}}
\newcommand{\Q}{\mathbb{Q}}
\newcommand{\R}{\mathbb{R}}
\newcommand{\N}{\mathbb{N}}
\newcommand{\Z}{\mathbb{Z}}
\newcommand{\SQ}{\Sigma_{\Q}}

\theoremstyle{definition}
\newtheorem{definition}{Definition}
\newtheorem{lemma}{Lemma}
\newtheorem{theorem}{Theorem}
\newtheorem{corollary}{Corollary}
\newtheorem*{theoremx}{Theorem}

\begin{document}

\begin{center}
{\LARGE\bfseries Algorithmically Presented Numbers and Canonical\\[2pt]
Representations in Cryptographic Protocols\par}

\vspace{10pt}
{\large Arslan Br\"omme\footnote{Dipl.-Inform., B.Sc., CISSP, CISA, CISM, CAISE.
\texttt{arslanb@chain-horizon.com}}\par}
\vspace{2pt}
Chain Horizon GmbH\par
\vspace{2pt}
Version 0.9.5.20 (Working Draft)\par

\vspace{8pt}
\fbox{\parbox{0.86\textwidth}{\centering\itshape
Preprint / working paper. This version is a work in progress and may be updated.
Comments are welcome. This document is an English translation of the corresponding
German version (v0.9.5.17) of this paper.}}
\end{center}

\vspace{6pt}
\noindent\textbf{Abstract}

\medskip
\noindent
This paper develops a representation-theoretic perspective on cryptographic
protocols. The focus is not solely on the computability of the abstract value as
an extensional property, but on the algorithmic structure of its presentation in
a representation system: for operational use in protocols, algorithmic
accessibility of the value does not suffice; in addition, its fixed presentation
is decisive. We distinguish three representation-theoretic notions ---
algorithmically approximable ($\Aapp$, the computable real numbers), finitely
exactly describable in a system ($\Afin(S)$), and canonical normalizability of a
system --- and show that there is no computable extensional canonicalizer that
uniformly transforms arbitrary approximation programs of computable real numbers
into unique finite value encodings. As the operational rational core
presentation we use the rational system with its canonical encoding
specification $\SQ$ (the fixed rules for valid fraction descriptions, canonical
codes, and normalization); the associated value set is $\Aex = \Q$. The notion
of a canonically serializable object class transfers this core idea to practical
protocol objects (concrete files as byte sequences, hash values, transaction
IDs, and normatively serialized payloads). We illustrate the consequences for
interoperability, well-definedness, and verification with fully worked-out toy
examples from symmetric and asymmetric encryption as well as hashing, and with a
real-world application example, the snaproot hash-anchoring protocol for
blockchain-based file integrity verification. The paper thereby shows that the
mathematical determinacy of a value and its operational uniqueness as a protocol
object are two different requirements. Once a normative representation
specification has been fixed, byte-level correctness and well-definedness
arguments can be carried out without further implementation-dependent
serialization or rounding decisions.


\section{Introduction and Motivation}

Real numbers are central objects of analysis and of many mathematical models. In
cryptography and theoretical computer science, however, there is a fundamental
tension between mathematical uniqueness and algorithmic reproducibility: real
numbers are uniquely defined as mathematical objects, but in general possess no
finite, canonical representation. Cryptographic protocols, by contrast, process
finite byte sequences. If these byte sequences are understood as representations
of semantic objects, then a uniquely specified encoding is helpful for
representation-invariant processing. If semantically equal objects are to
produce identical protocol bytes independently of their initial representation,
a canonical encoding or normalization is additionally required; we call this
stronger requirement \emph{byte-input invariance}.

Since this paper connects notions from computable analysis and from
cryptography, we first fix some basic terminology. A \emph{cryptographic
primitive} is a basic building block such as a hash function, an encryption
scheme, or a signature scheme. By a \emph{protocol} we mean a normative
specification of inputs, representations, and algorithmic processing steps. A
\emph{semantic object} is the abstract entity meant by a description (for
instance a number, a file, or a metadata record). \emph{Extensional} means:
depending only on the represented value, not on the concrete description or the
program text; \emph{intensional} or syntactic, conversely, is our term for
properties that depend on the chosen description, notation, or program.
Presentations are, in this sense, intensional or syntactic objects; the value
they denote constitutes the extensional level. \emph{Representation-invariant}
means: semantically equal presentations lead to the same specified operational
result. The stronger requirement that semantically equal presentations are
mapped to the same operational bit or byte sequence is what we call
\emph{byte-input invariance}. Byte-input invariance requires a canonical
encoding (or an equivalent normalization); general representation invariance can
also be achieved by directly extensional processing of the decoded value.
\emph{Uniform} means a single procedure that works for all inputs of a class (as
opposed to the mere existence of individual solutions for each input). An
\emph{object class} denotes, in what follows, a fixed class of semantic objects
to be encoded in a uniform manner, such as rational numbers, files as byte
sequences, hash values, transaction IDs, or structured payloads;
\emph{protocol objects} are those semantic or serialized objects that are
operationally processed in a protocol, such as key parameters, files, hash
values, transaction IDs, or payloads.

This paper addresses this discrepancy by distinguishing different forms of
algorithmic accessibility of real values and of their presentations. The goal is
not the introduction of new cryptographic primitives, but a structural
clarification of which classes of numbers are suitable as operational protocol
parameters and which occur primarily in analysis, modeling, and error
estimation.

The cryptographic literature implicitly assumes that keys, hash inputs, and
protocol parameters are given as well-defined, finitely representable objects.
The question of which classes of numbers themselves satisfy this property is,
however, rarely analyzed explicitly. In particular, it usually remains unstated
that real numbers, while uniquely defined as mathematical objects, possess no
canonical algorithmic representation without normative specifications. This
paper makes this implicit assumption explicit and examines the consequences for
cryptographic protocols. ``Canonical'' here denotes a uniquely fixed,
algorithmically reproducible representation which, once the encoding has been
fixed, requires no further implementation-dependent additional decisions.

The contribution of this paper does not lie in a new undecidability result ---
the notion of computable real numbers, the undecidability of their equality, and
the establishment of canonical serialization in cryptographic standards are
classical --- but in the systematic transfer of known representation-theoretic
limits to the specification of cryptographic protocol objects: in the joint
framing of computable analysis and cryptographic serialization, in the
distinction between analysis values and operational protocol representations,
and in the didactic illustration by means of encryption, hashing, hash
anchoring, and commitment semantics.

Concretely, the paper makes three contributions. First, it sharpens the notion
of an \emph{algorithmically presented number} for the context of cryptographic
protocol objects by explicitly separating the real value from its effective
presentation. Second, it formulates the classical undecidability of equality of
computable real numbers as a \emph{canonicalization barrier} for a universal
extensional value encoding (Theorem~1) and condenses its design consequence into
a design trilemma. Third, it transfers this perspective to cryptographic
protocols and proposes a taxonomy of operational representations that
distinguishes between exact byte identity, fixed serialization, and canonical
value presentation.

Figure~\ref{fig:pipeline} summarizes this perspective schematically: between the
abstract value, its description, a normal form, and the operational bit/byte
encoding lie separate representation decisions in each case; only the bit/byte
encoding constitutes the immediate input of the cryptographic protocol step.

\begin{figure}[ht]
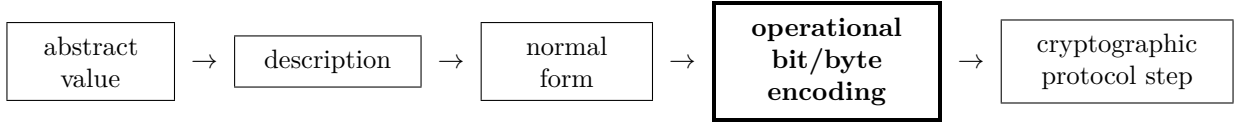

\centering
\small
\setlength{\fboxsep}{5pt}
\fbox{\parbox{1.85cm}{\centering abstract\\ value}}
$\;\rightarrow\;$
\fbox{\parbox{2.1cm}{\centering description}}
$\;\rightarrow\;$
\fbox{\parbox{1.9cm}{\centering normal\\ form}}
$\;\rightarrow\;$
{\setlength{\fboxrule}{1.2pt}\fbox{\parbox{2.6cm}{\centering\bfseries operational\\ bit/byte\\ encoding}}}
$\;\rightarrow\;$
\fbox{\parbox{2.7cm}{\centering cryptographic\\ protocol step}}
\caption{From the abstract value via the operational representation to
cryptographic processing. The highlighted bit/byte encoding already \emph{is}
the operational protocol object; the cryptographic protocol step processes it
but is not a further form of representation. At each upstream stage, missing
uniqueness, missing computability, or divergent encoding rules can lead to
different protocol bytes.}
\label{fig:pipeline}
\end{figure}

\section{Related Work}

The notion of computable real numbers goes back to Turing~\cite{turing1937}, who
gave a fundamental formal definition; it forms the foundation of modern
computability theory. A systematic treatment of effective analysis and of
representations of real numbers can be found in Weihrauch~\cite{weihrauch2000}.
Pour-El and Richards~\cite{pourel1989} study examples in which computable input
data lead to non-computable solutions of analytic problems, in particular for
differential equations; this underlines the strict inclusion
$\Aapp \subset \R$.

In constructive analysis~\cite{bishop1985}, mathematical existence is tied to
explicit construction procedures. The present paper takes up this idea but
focuses explicitly on canonical representability for algorithmic and
cryptographic protocols. The notion of a representation system used here is
deliberately more special and tailored to decidable languages of finite protocol
descriptions; it does not replace the more general theory of represented spaces
but isolates the subquestion relevant for canonical byte encodings. Modern
surveys of the general theory of represented spaces can be found, for instance,
in Pauly~\cite{pauly2016}; the computability of translations between
representations is there itself a separate, often non-trivial
problem~\cite{paulysteinberg2018} --- we deliberately consider the narrower case
of finite canonical codes. Practical systems of verified exact real
arithmetic~\cite{konecny2020} moreover show that real values can be processed
operationally via effective names and approximation procedures; this, however,
does not solve the question studied here of a unique finite value encoding.

Cryptographic primitives are traditionally defined over discrete, exactly
representable structures~\cite{goldreich, katzlindell}. Hash functions are
formalized as mappings of bit strings~\cite{bellarerogaway1993, fips180-4},
which emphasizes the importance of unique input representations.

The practical necessity of deterministic serialization is documented in detail
in individual standards: RFC~8785 (JSON Canonicalization
Scheme)~\cite{rfc8785} explicitly motivates canonical JSON by the fact that hash
and signature schemes require an invariant data format; RFC~8949
(CBOR)~\cite{rfc8949} gives rules for deterministic binary encoding, including
the stipulation of which number representation is to be used; ASN.1/DER~\cite{x690}
is an established canonical encoding rule. Less frequently, however, is this
necessity explicitly connected with the representation theory of computable real
numbers~\cite{turing1937, weihrauch2000}. It is precisely this
representation-theoretic perspective --- at the interface of
computability~\cite{turing1937, weihrauch2000} and the bit-string nature of
cryptographic primitives~\cite{goldreich, katzlindell} --- that the present
paper adopts.

\section{Algorithmically Presented Numbers}

\subsection{Basic Assumptions}

The following basic assumptions apply to the approximation, evaluation, and
normalization algorithms considered in Sections~3 and~4, on which the
computability statements rest. We use a model-independent notion of
deterministic algorithms. For a fixed approximation algorithm $M_x$ (which
approximates the value $x$), $n$ is the only input and precision parameter; the
value-dependent information is contained in the program code of $M_x$. In a
representation system, by contrast, $\mathrm{Eval}(d, n)$ explicitly relies on a
value-carrying finite description $d$. Oracles, randomness, and external system
states are not part of these computability statements; real cryptographic
protocols may use them in addition (for instance external blockchain states or
randomized commitments, cf.\ Sections~7.5 and~11).

\paragraph{Bit and byte convention.}
Finite bit strings $\{0,1\}^*$ serve throughout as the abstract encoding model.
Concrete serializations we understand as byte sequences over
$\{0,\dots,255\}^*$; via the usual eight-bit representation these can in turn be
modeled as bit strings. Where we formally write $\{0,1\}^*$, concrete byte or
ASCII encodings are thus included as a special case.

\subsection{Representation Systems}

The central notion of this paper is not a set of numbers but the manner in which
a real value is described. The computability of a value is an extensional
property of the value itself; its operational suitability for protocols, by
contrast, depends on the chosen presentation. The focus therefore lies on the
algorithmic structure of the presentation, not solely on the computability of
the value.

\paragraph{Terminology.}
In this paper, ``presentation'' (or ``description'') denotes the concrete finite
description $d$ itself by which a value is given; the \emph{algorithmically
presented number}, by contrast, is the pair $(\sem{d}, d)$. Here $\sem{d}$
denotes the semantic value represented by the description $d$ (the double
brackets read as ``the value denoted by $d$''). ``Representation'' or
``encoding'', by contrast, denotes the operational mapping of a semantic object
to a fixed form or byte sequence. A \emph{canonical representation} is such an
encoding with exactly one valid form per object.

\begin{definition}[Representation system]\label{def:repsys}
A representation system is a pair $S = (L, \mathrm{Eval})$, where
$L \subseteq \{0,1\}^*$ is a decidable language of finite descriptions and
$\mathrm{Eval}$ is a total algorithm (total meaning here: $\mathrm{Eval}$
terminates for all inputs of its specified domain, hence for all $d \in L$ and
$n \in \N$) that assigns to each $d \in L$ and $n \in \N$ a rational number
$\mathrm{Eval}(d, n)$ such that a uniquely determined value $\sem{d} \in \R$
exists with $|\mathrm{Eval}(d, n) - \sem{d}| < 2^{-n}$ for all $n$.
\end{definition}

\begin{definition}[Algorithmically presented number]\label{def:algnum}
An algorithmically presented number in $S$ is a pair $(\sem{d}, d)$ with
$d \in L$. Its value is $\sem{d} \in \R$, its presentation is $d$. The same real
value can be canonical and protocol-ready in one system, accessible only via
approximations in another, or represented by several semantically equivalent
descriptions.
\end{definition}

\paragraph{Notational convention.}
Where no confusion is possible, we write the value for short (e.g.\
``$k = 13/37$'') and understand the associated presentation as given along with
it.

\subsection{Three Representation-Theoretic Notions}\label{sec:threenotions}

We distinguish three interrelated representation-theoretic notions that are
located on different levels: a property of values (approximability), a relation
between a value and a representation system (finite describability), and a
property of an entire representation system (canonical normalizability). They
are not ordered within a single space of properties.

\paragraph{(i) Approximability.}
A value $x \in \R$ is \emph{algorithmically approximable} if there exists an
algorithm $M$ with $M(n) \in \Q$ and $|M(n) - x| < 2^{-n}$. The class $\Aapp$ of
these values is exactly that of the computable real numbers
(Turing~\cite{turing1937}, Weihrauch~\cite{weihrauch2000}); convergence is
effectively controlled by $n$. Such an algorithm $M$, which outputs for every
$n \in \N$ a rational approximation $M(n)$ with $|M(n) - x| < 2^{-n}$, we call a
\emph{fast rational Cauchy name} of the real value $x$.

\paragraph{(ii) Finite exact describability.}
A value $x$ is \emph{finitely exactly describable} in $S$ if there exists a
$d \in L$ with $\sem{d} = x$. The class is called $\Afin(S)$. It depends on the
system and can substantially exceed $\Q$ --- for instance, an algebraic system
(description = minimal polynomial and isolating interval~\cite{basupollackroy})
contains values such as $\sqrt{2}$. For a canonical representation, the
polynomial would additionally have to be normalized and the isolating interval
chosen according to a unique rule; a minimal polynomial together with an
arbitrary isolating interval is only an exact, not yet a canonical, description.

\paragraph{(iii) Canonical normalizability of the system.}
The system $S$ is called \emph{canonically normalizable} if there exists a total
computable algorithm $\mathrm{nf}\colon L \to L$ with
$\sem{\mathrm{nf}(d)} = \sem{d}$ and
$\mathrm{nf}(d_1) = \mathrm{nf}(d_2) \Leftrightarrow \sem{d_1} = \sem{d_2}$.
This is a property of the entire system, not of a subclass within $S$: if $S$ is
canonically normalizable, then every value represented by $S$ possesses a
computably determinable unique normal form.

\medskip
The three notions are genuinely distinct because they lie on different logical
levels: approximability is a property of values, finite describability a
relation between a value and a representation system, canonical normalizability
a property of representation systems. They therefore do not form nested
subclasses of one common system. A representation system in the sense of
Definition~\ref{def:repsys} has a decidable $L$ and a total $\mathrm{Eval}$; it
thus describes a uniformly computable family of real numbers. The totality of
all approximation programs, by contrast, does not form a representation system
in this sense: the set of programs that are total and satisfy the Cauchy
condition is not decidable; this follows from classical index-set results of
recursion theory (for instance Rice-style undecidability
theorems)~\cite{soare2016} or from halting-problem reductions; cf.\
\cite{weihrauch2000} for the computable-analysis context. We therefore treat
this totality as an external, non-decidable name space --- it assigns to every
computable real number at least one index, but is itself not a system with a
decidable description language. The negative statement of Section~4 concerns
exactly this name space: there is no computable extensional canonicalization of
arbitrary approximation programs.

\subsection{Canonical Encoding and the Operational Core Class}\label{sec:coreclass}

For protocols that require identical canonical bytes for semantically equal
objects, property~(iii) is the decisive system property, supplemented by the
requirement that the normal form is uniquely given in bytes. We capture this as
a canonical encoding.

\begin{definition}[Canonical encoding]\label{def:canenc}
A canonical encoding of an object class $K$ consists of a decidable language of
valid codes $C \subseteq \{0,1\}^*$, an encoder
$\mathrm{enc}\colon K \to C$ and a decoder $\mathrm{dec}\colon C \to K$, both
total computable with respect to the fixed representations, such that
$\mathrm{dec}(\mathrm{enc}(x)) = x$ holds for all $x \in K$ and, for each $x$,
the only valid representation is $\mathrm{enc}(x)$, i.e.\
$\mathrm{dec}^{-1}(\{x\}) = \{\mathrm{enc}(x)\}$. Canonicity is always to be
understood relative to a normatively fixed encoding specification.
\end{definition}

The computability of encoder and decoder is here always to be understood
relative to fixed effective input and output representations of the object
classes involved: an algorithm does not receive an abstract object directly, but
already a description of it (for instance a syntactically valid fraction for an
element of $\Q$).

\paragraph{Deterministic versus canonical.}
This distinction is central: a \emph{deterministic} serialization produces,
for a fixed input, reproducibly the same output. \emph{Canonicity} additionally
demands that every semantic object possesses exactly one valid encoding within
the specification. An encoder may, for instance, always output ``1/2''
(deterministic) while the associated decoder also accepts ``2/4'' --- the
procedure is then deterministic, but the code language is not canonical. The
expression $\mathrm{dec}^{-1}(\{x\})$ here denotes the set of all valid codes
that the decoder interprets as $x$; canonicity demands that this set is a
singleton for every $x$. That this distinction is practically relevant is shown,
for instance, by Base64: several distinct encodings can decode to the same data
and thereby break string uniqueness~\cite{chalkias2022}.

The last condition is the actual substance: it excludes that several valid byte
sequences denote the same object --- a problem that the mere function property
of $\mathrm{enc}$ does not capture. For $\Q$, one separates two languages for
this purpose: the language of all syntactically valid fraction descriptions
$L_{\Q} = \{\, p/q : p \in \Z,\ q \in \Z\setminus\{0\} \,\}$ and the language of
canonical codes
$C_{\Q} = \{\, p/q : q > 0,\ \gcd(|p|, q) = 1 \,\}$. The normalization
$\mathrm{nf}_{\Q}\colon L_{\Q} \to C_{\Q}$ is the computable reduction to lowest
terms with sign normalization; it satisfies
$\mathrm{nf}_{\Q}(\text{``1/2''}) = \mathrm{nf}_{\Q}(\text{``2/4''}) =
\mathrm{nf}_{\Q}(\text{``50/100''}) = \text{``1/2''}$, i.e.\ it maps all
descriptions of the same value to the same code. For $C_{\Q} \subseteq
\{0,1\}^*$ to actually fix a unique byte sequence per fraction
(Definition~\ref{def:canenc}), the lexical rules belong to the specification:
ASCII encoding, exactly one slash, the ASCII minus only as the sign of $p$, no
leading plus, no leading zeros (except for the digit 0 itself), no spaces,
positive denominator. For the rational value $r = p/q$ with $q > 0$ and
$\gcd(|p|, q) = 1$, concretely:
$\mathrm{enc}_{\Q}(r) = \mathrm{ASCII}(\mathrm{int}(p), \text{``/''},
\mathrm{nat}(q))$, where $\mathrm{int}$ denotes the normatively fixed decimal
syntax for integers with an optional minus sign and $\mathrm{nat}$ the
normatively fixed decimal syntax for positive natural numbers without a sign,
and $\mathrm{dec}_{\Q}(\text{``p/q''}) = p/q$ --- on the left stands a code (a
string), on the right the denoted rational number; then
$\mathrm{dec}_{\Q}(\mathrm{enc}_{\Q}(r)) = r$ and
$\mathrm{dec}_{\Q}^{-1}(\{r\}) = \{\mathrm{enc}_{\Q}(r)\}$. The injectivity
lemma in Section~4 merely shows that the encoder must be injective; that
different implementations choose \emph{the same} bytes requires the previously
jointly fixed specification $(C_{\Q}, \mathrm{enc}, \mathrm{dec})$.

The rational system $S_{\Q}$ with the specification $\SQ$ (valid descriptions
$L_{\Q}$, canonical codes $C_{\Q}$, normalization $\mathrm{nf}_{\Q}$) is
canonically normalizable. The set of its represented values is $\Q$. This
rational core presentation, the pair $(\Q, \SQ)$, we call for brevity the
\emph{operational core class}; its value set is $\Aex = \Q$
($= \Afin(S_{\Q})$). It is not
membership in $\Q$ but membership in the fixed presentation $(\Q, \SQ)$ that
guarantees the unique encoding --- the rational number $1/2$ could otherwise
appear as $2/4$, as a decimal number, or as an IEEE value. $\Aex = \Q$ is here a
deliberately chosen operational rational core class and \emph{not} the totality
of all finitely exactly describable or canonically presentable real numbers
($\sqrt{2}$, too, possesses suitable exact finite descriptions, cf.\
Section~\ref{sec:threenotions}(ii)). The operational uniqueness follows not from
mere membership in $\Q$ but from the pair $(\Q, \SQ)$.

The term ``exact'' here carries two meanings that must be distinguished: in
``finitely exactly describable'' (Section~\ref{sec:threenotions}(ii)) it merely
means that the value is described without approximation error --- an exact
description is not yet a canonical one. In the notation $\Aex$, by contrast, it
refers to the canonical determinability of the presentation. In both cases it
denotes no special status of the mathematical value. $\Q$ is also not the
\emph{simplest} canonically normalizable value set --- $\Z$, $\N$, or the dyadic
fractions are simpler, whereas algebraic systems (minimal polynomial and
isolating interval) are larger. Rather, $\Q$ is the smallest subfield of $\R$
and thus closed under addition, subtraction, multiplication, and division by
nonzero elements: the smallest operational core class within $\R$ with a
complete field structure. This motivates the choice $\Aex = \Q$ for arithmetic
protocol parameters.

The class boundary concerns uniform determinability, not the nameability of
individual values: a number such as $\pi$ can very well occur as an operational
parameter if a protocol defines the fixed identifier ``PI'' as a normative
symbol. What is impossible is merely to transform an arbitrary computable real
number, given solely by an arbitrary approximation procedure, uniformly into a
unique finite value encoding.

Terminologically, we distinguish \emph{normalization} and
\emph{canonicalization} as follows: a normalizer $\mathrm{nf}\colon L \to L$
works within a fixed representation system and maps valid descriptions to a
normal form of the same system; a canonicalizer denotes, more generally, a
procedure that maps presentations or names of a value to a unique
value-dependent code. In the rational system, the two perspectives practically
coincide; for arbitrary approximation programs of computable real numbers,
however, such a universal extensional canonicalization fails due to the
undecidability of value equality (Section~4).

Table~\ref{tab:notions} summarizes the representation-theoretic notions
introduced above and assigns them to their respective logical levels.

\begin{table}[ht]
\centering
\small
\begin{tabular}{|p{4.2cm}|p{3.2cm}|p{7.2cm}|}
\hline
\textbf{Notion} & \textbf{Level} & \textbf{Core statement} \\ \hline
$\Aapp$ & value & effectively rationally approximable \\ \hline
$\Afin(S)$ & value relative to $S$ & possesses a finite exact description in $S$ \\ \hline
canonically normalizable $S$ & system & value-equal descriptions receive the same normal form \\ \hline
$\Aex = \Q$ & value set & rational operational core class \\ \hline
$\SQ$ & specification & normative encoding and normalization rules for rational numbers \\ \hline
$(\Q, \SQ)$ & presented structure & rational values together with their normative encoding \\ \hline
canonically serializable object class & protocol level & exactly one valid encoding per object (Section~\ref{sec:serializable}) \\ \hline
\end{tabular}
\caption{The central representation-theoretic notions and their levels. The
notions deliberately lie on different logical levels and do not form a simple
hierarchy of sets of numbers.}
\label{tab:notions}
\end{table}

\subsection{Canonically Serializable Object Classes and Structural Properties}\label{sec:serializable}

\begin{definition}[Canonically serializable object class]\label{def:serializable}
An object class $K$ is called \emph{canonically serializable} with respect to a
specification $\Sigma$ if $\Sigma$ fixes a canonical encoding
$(C, \mathrm{enc}_{\Sigma}, \mathrm{dec}_{\Sigma})$
(Definition~\ref{def:canenc}) for all elements of $K$. The object $o \in K$ is
then given in \emph{$\Sigma$-conformant representation} by the code
$\mathrm{enc}_{\Sigma}(o)$. To avoid misunderstandings, we formally use
exclusively the term ``canonically serializable object class''; the informal
shorthand ``$\Aex$-compatible'' means nothing other than ``canonically
serializable with respect to a named specification'', and the analogy to $\Aex$
consists merely in the fact that such object classes, too, are encoded finitely,
uniquely, and reproducibly under a specification. The requirement concerns an
entire class under a fixed $(C, \mathrm{enc}_{\Sigma}, \mathrm{dec}_{\Sigma})$,
not a single object --- for a single object, canonical encoding would be
trivially satisfiable.
\end{definition}

Canonically serializable classes need not themselves be real numbers: concrete
files as byte sequences, hash values, transaction IDs, and normatively
serialized payloads share with $\Aex$ the property of canonical, finite,
reproducible representability under a fixed specification, provided that this
specification fixes a unique encoding per object. This transfers the operational
core idea from numbers to arbitrary protocol objects.

For the value classes, $\Aex = \Q \subsetneq \Aapp \subsetneq \R$ holds; both
inclusions are strict: $\pi \in \Aapp$ but $\pi \notin \Q$ ($\pi$ is computably
approximable but not rational), and non-computable real numbers lie in $\R$ but
not in $\Aapp$. On $\Aapp$ we consider the arithmetic operations
$(+, -, \times)$ inherited from $\R$ and division by nonzero elements; the usual
laws of arithmetic are inherited from $\R$, and what is non-trivial in each case
is closure.

For orientation: familiar constants such as $\pi$, $e$, or $\sqrt{2}$ are
computably approximable; non-computable real numbers can be constructed, for
instance, via binary sequences that encode an undecidable problem such as the
halting problem.

\paragraph{Remark (countability).}
Every algorithm possesses a finite description over a finite alphabet, and the
set of these descriptions is countable. A \emph{machine index} here denotes the
number of a program, or of a Turing machine, in a fixed effective numbering; in
what follows we use ``index'' as short for such a machine index. Since at least
one finite machine index is assigned to every computable real number, $\Aapp$ is
countable as the image of a countable set; several indices may describe the same
number.

\section{Limits of Canonical Normalization}\label{sec:limits}

Representation systems permit a precise statement about which presentations
qualify as operational parameters of executable protocols. The core is a limit:
there is no computable extensional canonicalization that uniformly transforms
arbitrary approximation programs of computable real numbers into unique finite
value encodings. It follows that protocols which are to map semantically equal
values to identical protocol bytes independently of their initial representation
must rely on a previously fixed, canonically normalizable presentation.

\begin{lemma}[Undecidability of equality on $\Aapp$]\label{lem:undecid}
There is no algorithm that decides, for two arbitrary elements
$x, y \in \Aapp$ given by approximating algorithms $M_x$, $M_y$, whether
$x = y$. This result is classical~\cite{weihrauch2000}; we sketch the proof for
self-containedness.
\end{lemma}

\begin{proof}[Proof sketch (reduction from the halting problem)]
Let $M$ be a Turing machine. Define the rational sequence $a_k = 2^{-t}$ if $M$
halts within $k$ steps (with $t \le k$ the halting step), and $a_k = 0$
otherwise. This sequence is computable, since it only requires simulating $M$
for $k$ steps; it is monotone and satisfies $|x_M - a_k| \le 2^{-k}$; by
outputting $a_{n+1}$ one obtains the required strict bound $< 2^{-n}$. The limit
is $x_M = 2^{-t}$ (if $M$ halts at step $t$) or $x_M = 0$ (if $M$ never halts).
Hence $x_M \in \Aapp$. We have $x_M = 0$ exactly if $M$ never halts. An equality
test ``$x_M = 0$?'' would thus decide the halting problem. Contradiction.
\end{proof}

The following result constitutes the formal limit of the approach. It is an
immediate consequence of the classical undecidability of equality of computable
real numbers (Lemma~\ref{lem:undecid}), formulated here as a limit of universal
extensional value canonicalization; the original contribution lies not in the
undecidability result itself but in this representation- and protocol-theoretic
reading.

\begin{theorem}[Canonicalization barrier --- no universal extensional
canonicalization]\label{thm:barrier}
Let $I_{\mathrm{app}}$ be the set of machine indices that produce fast rational
Cauchy names of computable real numbers (index $e$ denotes the value $x_e$);
these are exactly the approximating algorithms considered in
Lemma~\ref{lem:undecid}, now regarded as input objects. There exists no
partially computable function $c\colon \N \rightharpoonup \{0,1\}^*$ that is
defined on all $e \in I_{\mathrm{app}}$ and there satisfies
$c(e) = c(f) \Leftrightarrow x_e = x_f$. (On invalid indices, $c$ may diverge;
the validity of an index need not be decidable.)
\end{theorem}

\begin{proof}
Such a function $c$ takes the machine indices themselves as input and produces
the same finite code for value-equal programs --- including syntactically
different ones (extensionality). If $c$ existed, the finite string comparison
$c(e) = c(f)$ would decide the equality of two arbitrary computable real numbers
given by indices, contradicting Lemma~\ref{lem:undecid}. What is decisive is the
effective availability of $c$ on the indices: the mere existence of
value-canonical descriptions does not suffice, since it provides no algorithm
that constructs them from the given programs.
\end{proof}

\paragraph{Interpretation.}
The barrier does not concern the existence of finite names for individual
computable numbers --- every single one possesses such names. It concerns the
uniform, computable, and extensional derivation of a unique code from arbitrary
approximation programs. Equivalently formulated, a \emph{design trilemma}
results: a protocol over arbitrary computable real numbers cannot
simultaneously (1)~accept arbitrary valid approximation programs as inputs,
(2)~treat syntactically different programs for the same value identically
(value-equal programs receive the same code), and (3)~map value-distinct
programs computably to distinct finite codes. Conditions (2) and (3) capture
exactly the two directions of $c(e) = c(f) \Leftrightarrow x_e = x_f$; at least
one of these three requirements must be given up or restricted --- the design
consequences of this are discussed in Section~9.

\begin{corollary}[No computable universal translation; operational core
class]\label{cor:notrans}
There exists no representation system $S$ that represents all values from
$\Aapp$ by finite canonical codes, together with a partially computable
extensional translation, defined on all valid approximation programs, that maps
each program to the canonical code of its value in $S$. (Individual systems with
finite canonical value descriptions do exist --- for instance $S_{\Q}$ for
$\Q$; what is excluded is only the pair consisting of such a system and an
effective translation covering all of $\Aapp$.) If a protocol is to map
semantically equal values to identical protocol bytes independently of their
initial representation, it must therefore restrict its input interface to a
previously fixed, canonically normalizable presentation or provide an equivalent
additional representation rule, instead of operating directly on all of
$\Aapp$. The choice $S_{\Q}$ yields the rational core presentation $(\Q, \SQ)$
with value set $\Aex = \Q$ --- the smallest core class within $\R$ with a
complete field structure; larger ones are possible (for instance algebraic
systems) but require a correspondingly richer normative specification.
\end{corollary}

\begin{proof}
A computable universal translation would be an extensional canonicalizer in the
sense of Theorem~\ref{thm:barrier} and therefore cannot exist. For the fixed
system $S_{\Q}$, by contrast, normalization is effectively computable by
computing the greatest common divisor and normalizing the sign; the set of its
represented values is $\Q = \Aex$.
\end{proof}

\begin{lemma}[Injectivity lemma for reproducible implementations]\label{lem:inject}
Let $\mathcal{E}\colon K \to Y$ be an injective evaluation function on a
parameter class $K$. If a reproducibly correct computable implementation
$(\mathrm{enc}, \hat{E})$ of $\mathcal{E}$ exists --- that is,
$\mathrm{enc}\colon K \to \{0,1\}^*$ and
$\hat{E}\colon \{0,1\}^* \to Y$ are total and computable, and
$\hat{E}(\mathrm{enc}(o)) = \mathcal{E}(o)$ holds for all $o \in K$ --- then
$\mathrm{enc}$ is total, computable, finite, and injective.
\end{lemma}

\begin{proof}
Totality and computability of $\mathrm{enc}$ hold by assumption, finiteness of
the output by the codomain $\{0,1\}^*$. For injectivity, let
$\mathrm{enc}(o_1) = \mathrm{enc}(o_2)$. Since $\hat{E}$ is a function, it
follows that $\hat{E}(\mathrm{enc}(o_1)) = \hat{E}(\mathrm{enc}(o_2))$, hence by
the correctness condition $\mathcal{E}(o_1) = \mathcal{E}(o_2)$, and by the
injectivity of $\mathcal{E}$ finally $o_1 = o_2$.
\end{proof}

The lemma yields only the injectivity of $\mathrm{enc}$, not yet a canonical
encoding in the sense of Definition~\ref{def:canenc}: if, in addition, the code
language $C = \mathrm{enc}(K)$ is decidable and a computable inverse decoder
$\mathrm{dec}$ is given, then $(C, \mathrm{enc}, \mathrm{dec})$ forms a
canonical encoding. The uniqueness of the valid byte sequence is thus an
additional property to be fixed normatively, not a corollary of mere
injectivity.

The lemma isolates a necessary condition for a correct implementation of
injective abstract evaluations: the operational encoding must not identify
objects that the abstract evaluation $\mathcal{E}$ distinguishes.
Equivalently --- via the separation of the abstract evaluation $\mathcal{E}$ and
the implemented evaluation $\hat{E} \circ \mathrm{enc}$ --- the encoding must
inherit the injectivity of the abstract function; this is a necessary (not yet
sufficient) condition of canonical encoding.

\paragraph{Scope (representation layer vs.\ hash layer).}
The injectivity requirement does not concern cryptographic hash functions
themselves: due to their finite codomain, these cannot be injective, and
collision resistance merely means that collisions are practically hard to find.
The requirement rather concerns the upstream representation layer: if
semantically equal objects are to produce the same hash or signature input
independently of their initial representation, they must first be mapped to the
same byte sequences by a fixed canonicalization. The lemma is here applied not
to the hash function but to an upstream, value-distinguishing representation
mapping; the hash function processes only its byte output. Its collision
resistance is a cryptographic property separate from this. The lemma thus
concerns the well-definedness of the generation and comparison steps, not the
security of the hash layer.

With this, the question raised in the introduction --- which classes of numbers
qualify as operational protocol parameters --- is answered on the level of
computability. The statement is to be phrased with care: a protocol can
certainly operate in a well-defined manner on byte sequences even if several byte
sequences denote the same semantic entity --- for instance, by having a decoder
accept different notations and subsequently operate on the decoded value. A
canonicalization is required, however, if the protocol is to represent semantic
object equality by equality of the serialized byte sequences, or to map
semantically identical objects to the same hash or signature input independently
of their initial representation. For this class of protocols, the rational core
presentation $(\Q, \SQ)$ or a canonically serializable object class provides an
immediately available solution.

\section{Symmetric Encryption: A Comparison (Stream Cipher via XOR)}

\subsection{A Key in Rational Core Presentation (Fully Worked Example)}\label{sec:xorexample}

We consider a stream-cipher construction based on XOR. Let the key be
$k = 13/37 \in \Aex$. From the canonical rational normal form $13/37$, the local
toy example extracts the mathematical integer pair
$\mathrm{pair}_{\mathrm{toy}}(13/37) = (p, q) = (13, 37)$; this is not the
general specification $\SQ$. The following generator works directly on this pair
$(p, q)$ --- a concrete byte encoding of the pair is not needed here. The
keystream is generated deterministically as a byte sequence
$S = (s_0, \dots, s_{L-1})$. The following generator is deliberately not
cryptographically secure; it serves exclusively to make the representation
question visible in a fully traceable example:
\begin{align*}
s_i &= (p + 31 \cdot q + i) \bmod 256\\
p &= 13,\quad q = 37\\
p + 31 \cdot q &= 13 + 31 \cdot 37 = 13 + 1147 = 1160
\end{align*}
\begin{center}
\begin{tabular}{l}
$i = 0$: $s_0 = 1160 \bmod 256 = 136$ \texttt{(0x88)}\\
$i = 1$: $s_1 = 1161 \bmod 256 = 137$ \texttt{(0x89)}\\
$i = 2$: $s_2 = 1162 \bmod 256 = 138$ \texttt{(0x8A)}\\
$i = 3$: $s_3 = 1163 \bmod 256 = 139$ \texttt{(0x8B)}\\
$i = 4$: $s_4 = 1164 \bmod 256 = 140$ \texttt{(0x8C)}\\
$S$ (hex) = \texttt{[88, 89, 8A, 8B, 8C]}
\end{tabular}
\end{center}
Let the plaintext $P$ be ``HALLO'' in ASCII (hexadecimal):
$P = \texttt{[48, 41, 4C, 4C, 4F]}$. Encryption is performed bytewise by XOR
with the keystream: $C_i = P_i \oplus s_i$.
\begin{center}
\begin{tabular}{l}
\texttt{48} $\oplus$ \texttt{88} = \texttt{C0}\\
\texttt{41} $\oplus$ \texttt{89} = \texttt{C8}\\
\texttt{4C} $\oplus$ \texttt{8A} = \texttt{C6}\\
\texttt{4C} $\oplus$ \texttt{8B} = \texttt{C7}\\
\texttt{4F} $\oplus$ \texttt{8C} = \texttt{C3}\\
$C$ (hex) = \texttt{[C0, C8, C6, C7, C3]}
\end{tabular}
\end{center}
Decryption proceeds identically: $P = C \oplus S$. Reproducibility is guaranteed
because the rational normal form $13/37$ is uniquely fixed according to $\SQ$
and the generator operates deterministically on the mathematical pair $(p, q)$
extracted from it. (If $k$ is instead to be transmitted or hashed as a byte
sequence, a byte encoding of the pair must additionally be fixed --- with tuple
delimitation and integer encoding; cf.\ the separate byte variant
$\mathrm{bytes}_{\mathrm{toy}}$ in Section~\ref{sec:toyhashenc}.)

\subsection{A Comparison: A Real Key Without a Normatively Fixed Presentation}\label{sec:xorcontrast}

For comparison, we consider the same approach with a key for which no normative
presentation is prescribed, say $k = \pi$, where the protocol is given only an
approximation procedure but no fixed rounding or precision. It is not the value
$\pi$ that is the problem, but its unfixed presentation: for the generator $G$
to work, $k$ must be transformed into a finite byte sequence, e.g.\ as a decimal
string (``3.14'', ``3.1416'') or as a floating-point number (IEEE~754
float/double~\cite{ieee754}). Within a fixed IEEE~754 format and for finite
non-NaN values, the decoding of a given bit string is indeed unique; the
non-uniqueness arises, however, in the implementation-dependent choice of the
mapping of the intended parameter to a finite approximation, which additionally
depends on format, precision, and rounding rule (for instance $\pi$ in single
vs.\ double precision, or after different rounding) and yields different bit
strings. This representation is therefore not canonical without additional
stipulations.

\section{Asymmetric Encryption: A Comparison (Toy RSA)}

\subsection{RSA with a Fixed Integer Encoding (Fully Worked Example)}\label{sec:rsaexample}

All mathematical RSA values lie in $\Z \subset \Q$. Their operational
reproducibility, however, does not follow from this set-theoretic inclusion but
additionally presupposes a fixed encoding of the integers and key structures
(the rational core presentation $(\Q, \SQ)$ or a corresponding integer
specification). This classification is not meant to redefine RSA but to
emphasize that RSA operates exclusively with numbers that are finitely and ---
given a fixed encoding --- canonically representable. The following parameters
are deliberately small and serve illustration only (toy).
\begin{center}
\begin{tabular}{l}
$p = 11$\\
$q = 17$\\
$n = p \cdot q = 187$\\
$\varphi(n) = (p - 1)(q - 1) = 10 \cdot 16 = 160$\\
$e = 7$\\
Determine $d$ with $e \cdot d \equiv 1 \pmod{160}$:\\
$7 \cdot 23 = 161 \equiv 1 \pmod{160} \Rightarrow d = 23$\\
Public key: $(n, e) = (187, 7)$\\
Private key: $d = 23$
\end{tabular}
\end{center}
Encryption of a block $m = 42$ ($< n$): $c = m^e \bmod n$. Decryption:
$m = c^d \bmod n$.
\begin{center}
\begin{tabular}{l}
Encryption: $c = 42^7 \bmod 187 = 15$\\
Decryption: $m = 15^{23} \bmod 187 = 42$
\end{tabular}
\end{center}
All mathematical intermediate values are exact and deterministically
reproducible; an operational representation reproducible in all
specification-conformant implementations additionally presupposes the
aforementioned integer and key-structure encoding. Concretely: the modulus
$n = 187$ becomes a unique byte sequence only through a fixed integer encoding
rule. Under the rule ``unsigned, big-endian, fixed length 2 bytes'' we have
\begin{center}
$187 = \texttt{0x00BB} \rightarrow \texttt{[0x00, 0xBB]}$ \ (big-endian, 2
bytes, unsigned)
\end{center}
Without such a rule it would remain open whether $187$ is serialized as one byte
(\texttt{[0xBB]}), as two bytes in big- or little-endian
(\texttt{[0x00, 0xBB]} vs.\ \texttt{[0xBB, 0x00]}), or as an ASCII digit string
(``187'' = \texttt{[0x31, 0x38, 0x37]}). Thus even exact integers do not by
themselves determine concrete protocol bytes; endianness, length, and sign rule
belong to the operational specification.

\subsection{A Comparison: Real Numbers as Key Material (the $\R \to \Z$ Problem)}\label{sec:rsacontrast}

An approach that derives RSA keys from real numbers requires a mapping
$\R \to \Z$ (e.g.\ via decimal digits, rounding, or hashing of an
approximation). Without canonical rules, different approximations of the
intended real parameter can and in general will lead to different integers and
hence to different keys (with rounding, two approximations may also yield the
same integer). Key generation without additional encoding or specification rules
is therefore implementation-dependent and not well-defined.

\section{Hash Functions: A Comparison (Toy Hash)}

\subsection{A Toy Hash Function (Fully Computable)}

Since cryptographic hash functions such as SHA-256 are too large for a manual
worked example, we use a fully computable toy hash function to make the
representation question visible (no security):
\[
H(x) = \Bigl(\ \sum_{i=0}^{L-1} (i + 1) \cdot x_i\ \Bigr) \bmod 256,
\quad\text{where } x \text{ is a byte sequence.}
\]

\subsection{Rational Normal Form with a Fixed Local Toy Byte Encoding}\label{sec:toyhashenc}

Let $k = 13/37 \in \Aex$ and let a local toy byte encoding be
$\mathrm{bytes}_{\mathrm{toy}}(13/37) = [13, 37]$ (decimal as bytes; this byte
mapping is to be distinguished from the pair mapping
$\mathrm{pair}_{\mathrm{toy}}$ of Section~\ref{sec:xorexample}). This local toy
byte encoding suffices only for the example: it encodes one byte each for
numerator and denominator and is unsuitable as a general encoding of rational
numbers, since numerator and denominator can be arbitrarily large and negative
(a general encoding requires a sign rule, variable lengths, and unique tuple
delimitation). Then:
\begin{align*}
H([13, 37]) &= (1 \cdot 13 + 2 \cdot 37) \bmod 256\\
&= (13 + 74) \bmod 256\\
&= 87
\end{align*}
The hash value is uniquely determined because the input is uniquely fixed.

\subsection{A Comparison: A Real Value (Representation Dependence)}\label{sec:hashcontrast}

Let $k = \pi$ without a normatively fixed encoding. Two plausible finite
approximations intended to serve as operational inputs for the intended
parameter $\pi$ are then the ASCII strings ``3.14'' and ``3.1416''. They
represent different rational values ($157/50$ and $3927/1250$, respectively);
the resulting byte sequences differ and produce different hash values --- the
cause is not the value $\pi$ but the absence of a fixed presentation.
\begin{align*}
\text{``3.14''} &\rightarrow [51, 46, 49, 52]\\
H &= (1 \cdot 51 + 2 \cdot 46 + 3 \cdot 49 + 4 \cdot 52) \bmod 256\\
&= (51 + 92 + 147 + 208) \bmod 256 = 498 \bmod 256 = 242\\[4pt]
\text{``3.1416''} &\rightarrow [51, 46, 49, 52, 49, 54]\\
H &= (1 \cdot 51 + 2 \cdot 46 + 3 \cdot 49 + 4 \cdot 52 + 5 \cdot 49 + 6 \cdot 54) \bmod 256\\
&= 1067 \bmod 256 = 43
\end{align*}
The hash here depends on the chosen operational byte sequence and not directly
on the intended abstract parameter $\pi$. The same holds for an IEEE~754
encoding: within a fixed format and for finite non-NaN values, the decoding of a
given bit string is unique, but the mapping of the intended irrational parameter
to such a bit string additionally depends on format, precision, and rounding
rule (single vs.\ double precision, different rounding) and yields different
byte sequences and hence different hash values. This motivates a normative
representation rule and --- where byte-input invariance is demanded --- a
canonical representation or equivalent normalization.

\subsection{Transfer to Real SHA-256}\label{sec:sha}

The representation dependence made visible with the toy hash occurs unchanged
with real hash functions. Here, SHA-256(``3.14'') denotes the hash of the
ASCII/UTF-8 byte sequence \texttt{33 2e 31 34} (without quotation marks), and
SHA-256($[13, 37]$) the hash of the two raw bytes \texttt{0d 25}. For SHA-256:
\begin{center}
\begin{tabular}{l}
SHA-256(``3.14'') = \texttt{2efff126\ldots43500215}\\
SHA-256(``3.1416'') = \texttt{23741492\ldots75a42602}
\end{tabular}
\end{center}
Both strings are plausible finite approximations that could be used as
operational inputs for the intended parameter $\pi$. Since they represent
different rational values and different byte sequences, they yield different
hash values. For the fixed local toy encoding, by contrast, the input byte
sequence is uniquely determined:
\begin{center}
SHA-256($[13, 37]$) = \texttt{51698d9a\ldots2f94f015}
\end{center}
The example does not illustrate the internal workings of SHA-256, but the fact
that real hash functions hash concrete byte sequences --- not abstract
mathematical objects. (The full 256-bit values are given in
Appendix~\ref{app:hash}.)

The examples from symmetric encryption, RSA, and hashing jointly show that the
mathematical determinacy of a value does not guarantee operational byte
uniqueness; the latter arises only through fixed encoding, serialization, or
canonical normalization.

\subsection{A Real-World Hash-Anchoring Protocol: snaproot}\label{sec:snaproot}

Before the example, it is helpful to separate three related but distinct
situations --- a small taxonomy of operational representations that otherwise
easily merge in discussion: (a)~\emph{exact byte identity} --- an object is
exactly one byte sequence $F \in \{0,1\}^*$, which is hashed as such; no further
canonicalization is needed, since the protocol deliberately commits to the
bytes. (b)~\emph{Fixed serialization of structured objects} --- a structured
object is transformed into bytes according to a fixed rule. Within (b) we
distinguish further: deterministic or explicitly specified (the rule reproducibly
produces the same bytes from a given structured value) and additionally
canonical (there exists exactly one valid encoding per semantic object, so that
different initial representations of the same object also produce the same
bytes). (c)~\emph{Canonical value presentation of real numbers} --- several
descriptions of the same real value are to be reduced to the same normal form.
snaproot combines two of these levels: the file hash is a case of (a), the
structured SR1 memo payload a case of (b), namely --- as elaborated below --- in
the sense of a prescribed, not necessarily canonical, serialization.
Table~\ref{tab:taxonomy} summarizes this taxonomy and assigns to each of the
three levels its reference object and an example.

\begin{table}[ht]
\centering
\small
\begin{tabular}{|p{5.2cm}|p{3.6cm}|p{3.4cm}|}
\hline
\textbf{Level} & \textbf{Reference object} & \textbf{Example} \\ \hline
(a) exact byte identity & concrete bytes & file $F$ \\ \hline
(b) fixed serialization & structured data & SR1 payload \\ \hline
(c) canonical value presentation & semantic value & reduced fraction \\ \hline
\end{tabular}
\caption{Three operational representation levels. Levels (a)--(c) are different
requirements, not stages of a ranking.}
\label{tab:taxonomy}
\end{table}

The examples so far are deliberately small. The same representation problem,
however, also occurs in real cryptographic protocols. One example is
snaproot~\cite{snaproot}, a blockchain-based hash-anchoring protocol for file
integrity verification; it is the author's own case study, not independent
external evidence. Insofar as the term ``commitment'' is used below, it is meant
in the practical protocol sense of a public, later
verifiable commitment to a concrete file byte sequence or its hash (a public hash
commitment); the classical hiding property of cryptographic commitment schemes
is not claimed thereby. A file is treated there not as an abstract mathematical
object but as a finite byte sequence $F \in \{0,1\}^*$. The hash algorithm
SHA-256 deterministically maps this byte sequence to a 256-bit value $H$:
\[
H = \text{SHA-256}(F)
\]
What is hashed is not the semantic content, not the file name, and not the
visual rendering, but exactly the concrete byte sequence of the file.

The value $H$ is subsequently anchored on the blockchain in a structured SR1
payload with a prescribed format. The SR1 payload is the memo payload serialized according
to the snaproot SR1 format, anchored as the memo content of a transaction; the
terms SR1 structure and SR1 parser rule refer in what follows to the layout of
this payload and to the rule for its unique evaluation. ``SR1'' is
snaproot-specific format terminology. The payload is a case of fixed
serialization of structured data (version prefix, hexadecimal hash, client
timestamp, and further metadata fields); the user wallet is here the
\emph{signer} of the transaction, not a field of the payload:
\begin{align*}
\mathrm{payload} &= \mathrm{enc}_{\mathrm{SR1}}(H, \mathrm{timestamp},
\mathrm{metadata})\\
\mathrm{txid} &= \mathrm{Anchor}_{\mathrm{wallet}}(\mathrm{payload})
\end{align*}
The later verification is likewise a purely finite algorithm:
\begin{center}
\begin{tabular}{l}
$\mathrm{Verify}(F, \mathrm{txid})$:\\
\quad $H' = \text{SHA-256}(F)$\\
\quad $H_{\mathrm{chain}} = \mathrm{ExtractHash}(\mathrm{txid})$\\
\quad \textbf{return} $H' = H_{\mathrm{chain}}$
\end{tabular}
\end{center}
All relevant objects here possess finite and reproducible representations: the
file $F$ as a byte sequence (exact byte identity), the hash $H$ as a 256-bit
string, the transaction ID $\mathrm{txid}$ as a finite character string, and the
SR1 payload as an explicitly specified structured string; the wallet address acts as
the signer of the transaction. In the terminology of this paper, the operational
level of the protocol therefore lies entirely on the side of fixed finite
representations --- the hash input is the exactly fixed byte sequence of the
file (exact byte identity on the input side), the SR1 payload a case of fixed
serialization of structured data. (Whether the SR1 payload is canonical in the
strong sense of Definition~\ref{def:canenc} --- exactly one valid encoding per
metadata object --- depends on uniqueness rules for separators, GPS and null
values, capitalization, etc., which the fixed pipe format does not necessarily
determine completely; we therefore speak of a payload whose format is prescribed
but not necessarily canonical.) The snaproot example thus shows that real hashing and
verification protocols require a uniquely specified operational
representation --- depending on the object level, either exact byte identity or
a fixed serialization of structured data.

snaproot thereby illustrates two different operational levels: the file is
committed as a concrete byte sequence, while the metadata enter the payload via
a fixed structured serialization. A direct commitment to an abstract real value, by
contrast, would be well-defined only after fixing a representation or
approximation rule.

\section{Representation, Well-Definedness, and Correctness Arguments}\label{sec:correctness}

Beyond their operational significance for cryptographic protocols, fixed
algorithmic presentations also have a conceptual advantage for correctness
arguments: a fixed operational presentation turns the algorithm executed on it
into a uniquely determined function, so that elementary correctness proofs
become possible without representation-dependent auxiliary assumptions. The
correctness proof for XOR in Section~\ref{sec:xorproof} does not directly
illustrate the injectivity lemma (Lemma~\ref{lem:inject}) of
Section~\ref{sec:limits} --- it requires no injectivity of the keystream
generator --- but this more general point: with a fixed encoding of the key, the
generator is deterministic and the procedure uniquely determined.

The further connection --- that byte-input-invariant processing requires a
canonical encoding --- concerns the verification and comparison steps
(Section~\ref{sec:verifwd}), not the XOR correctness proof itself. The
counterexamples in Sections~\ref{sec:xorcontrast}, \ref{sec:rsacontrast},
and~\ref{sec:hashcontrast} show what goes
wrong without a fixed presentation: different operational approximations of the
same intended real parameter lead to different byte sequences and hence to
non-reproducible results.

\paragraph{Three properties to be separated.}
Throughout, we distinguish three different requirements: \emph{correctness} (a
procedure yields the specified result --- for this, a deterministic encoding
already suffices), \emph{byte-input invariance} (semantically equal
presentations produce the same operational byte input --- for this, a canonical
encoding is required; the more general representation invariance of the same
result can also be achieved without byte equality by extensional processing),
and \emph{cryptographic security} (secrecy, collision or second-preimage
resistance). None of these properties implies another; in particular, a
canonical encoding does not make a toy generator or a toy hash construction
secure. That the standard security of a signature primitive (existential
unforgeability) does not automatically guarantee all properties expected by the
surrounding protocol layer is also shown by the applied security literature ---
for instance regarding the unique binding of signatures to public
keys~\cite{seemslegit}.

\subsection{Abstract Encryption Step}

Given a message $m \in \{0,1\}^*$ and a deterministic keystream generator
$G\colon \{0,1\}^* \times \N \to \{0,1\}^*$ with $|G(b, L)| = L$, which from an
encoding $b$ and a length $L$ produces a keystream of exactly this length. With
$S_{k,m} = G(\mathrm{enc}(k), |m|)$, the XOR-based encryption is (bitwise; $m$
and $S_{k,m}$ are bit strings of equal length):
\begin{align*}
\mathrm{Enc}(k, m) &= m \oplus S_{k,m}, \qquad S_{k,m} = G(\mathrm{enc}(k), |m|)\\
\mathrm{Dec}(k, c) &= c \oplus G(\mathrm{enc}(k), |c|)
\end{align*}

\subsection{Correctness Proof for a Fixed Key Encoding}\label{sec:xorproof}

Let $K$ be a key class with a fixed deterministic operational encoding
$\mathrm{enc}\colon K \to \{0,1\}^*$ and let $G$ be the above keystream
generator. Then $S_{k,m} = G(\mathrm{enc}(k), |m|)$ is, for every $k \in K$ and
every message $m$, well-defined, deterministic, and of length $|m|$. For XOR
correctness this determinacy suffices; neither injectivity of $G$ nor semantic
canonicity is required. Even if two different keys produced the same keystream,
$(m \oplus S) \oplus S = m$ would remain correct.

\begin{theoremx}[Correctness]
For every key class $K$ with a fixed deterministic encoding, for all $k \in K$
and all messages $m$: $\mathrm{Dec}(k, \mathrm{Enc}(k, m)) = m$.
\end{theoremx}

\begin{proof}
\begin{enumerate}
\item $S_{k,m} = G(\mathrm{enc}(k), |m|)$ is uniquely determined and
$|S_{k,m}| = |m|$.
\item $c = \mathrm{Enc}(k, m) = m \oplus S_{k,m}$, hence $|c| = |m|$.
\item $\mathrm{Dec}(k, c) = c \oplus G(\mathrm{enc}(k), |c|) = c \oplus S_{k,m}$.
\item $(a \oplus b) \oplus b = a$ (involutivity of XOR, bitwise).
\item $\Rightarrow \mathrm{Dec}(k, c) = (m \oplus S_{k,m}) \oplus S_{k,m} = m$.
\end{enumerate}
Thus $\mathrm{Dec}(k, \mathrm{Enc}(k, m)) = m$.
\end{proof}

A canonical encoding becomes additionally relevant here only when semantically
equal key presentations are to produce the same keystream independently of their
initial notation --- that is, when representation invariance rather than mere
determinacy is demanded. The rational core presentation $(\Q, \SQ)$ provides a
concrete instance of such an encoded key class; $\Q$ appears here not as a
mathematical prerequisite of XOR correctness but as a concrete choice with an
immediately available canonical encoding.

\subsection{Well-Definedness of the snaproot Verification}\label{sec:verifwd}

The snaproot hash-anchoring protocol provides a real instance of the same
principle. The verification function is well-defined because its inputs are
given as uniquely specified finite representations: $F$ as a concrete byte
sequence and the payload according to the fixed SR1 structure.

\begin{theoremx}[Well-definedness]
Let $F$ be a finite byte sequence and $\mathrm{txid}$ a transaction ID that
refers to a uniquely retrievable on-chain memo payload with stored hash
$H_{\mathrm{chain}}$. Here the blockchain network, a finalized ledger view ---
that is, a stable, confirmed on-chain state --- and the SR1 parser rule are
assumed to be fixed. Then $\mathrm{Verify}(F, \mathrm{txid})$ yields a
deterministic result that is, under these fixed assumptions and for an identical
byte sequence, reproducible independently of the implementation.
\end{theoremx}

\begin{proof}[Proof idea]
Since $F$ is a finite byte sequence, SHA-256($F$) is uniquely determined. Since
$\mathrm{txid}$ is a finite character string and the payload structure is fixed,
$\mathrm{ExtractHash}(\mathrm{txid})$ is also uniquely determined. Verification
thus reduces to the comparison of two finite bit strings. Equality of finite bit
strings is decidable. Consequently, the result TRUE or FALSE is well-defined.
\end{proof}

The result TRUE here denotes the equality of the two SHA-256 values, not an
information-theoretic proof of the equality of the files: from
$H' = H_{\mathrm{chain}}$ it does not follow mathematically that the verified
file is byte-identical to the originally anchored file, since hash collisions
are possible in principle. The interpretation as an integrity proof rests on the
cryptographic assumption of the second-preimage or collision resistance of
SHA-256 --- a property separate from the representation and well-definedness
question treated here.

For an abstract real object $x \in \R$, this statement does not hold without an
additional encoding. Since $\R$ is uncountable but $\{0,1\}^*$ is countable, no
injective finite encoding of all real numbers exists; operational protocols must
therefore either be restricted to a canonically encodable subclass
$K \subseteq \R$ or fix an explicit approximation or canonicalization.

It is important to separate this cardinality argument from the result of
Section~\ref{sec:limits}: it concerns all of $\R$ ($|\R| > |\{0,1\}^*|$). For
the countable class $\Aapp$, cardinality is not the obstacle --- here, instead,
the impossibility of a computable extensional canonicalization from
Theorem~\ref{thm:barrier} applies. Both arguments lead to the same design
consequence but rest on different foundations.

Only an encoding fixed for the respective protocol purpose --- injective where
value distinguishability is demanded and canonical where byte-input invariance
is demanded --- transforms the elements of an operationally used subclass $K$
into operational protocol bytes. Different choices of encoding in general lead
to different commitments. Precisely herein lies the role of canonically
serializable representations (Definition~\ref{def:serializable}): they make
hashing, anchoring, and verification mathematically and practically well-defined
operations as soon as byte-input-invariant processing is demanded.

\subsection{Why the Operational Instantiation Is Not Well-Defined Without a
Fixed Presentation}

Now let the key $k$ be a real value, say $k = \pi$. To execute $G(k)$, $k$ must
be transformed into a finite representation (e.g.\ a decimal string or an
IEEE~754 floating-point number). Without normative specifications (format,
rounding, precision, normal form), this representation is not unique. Here, too,
the non-uniqueness lies not in the decoding of a concrete bit string within a
fixed IEEE~754 format, but in the mapping of the irrational value to a finite
approximation. The operational instantiation of $G(k)$ is therefore not
well-defined. The premises of the correctness theorem of
Section~\ref{sec:xorproof} are thus not fulfilled as long as no representation
rule has been fixed.

\subsection{Conclusions}

\begin{itemize}
\item A fixed effective representation makes the operational inputs and the
protocol steps defined on them uniquely determined; canonical representations
moreover guarantee byte-input invariance. Under deterministic further
processing, representation invariance follows from this.
\item Once the representation specification has been fixed, correctness proofs
can be carried out without further unstated or implementation-dependent
representation assumptions (e.g.\ divergent floating-point formats or rounding
rules).
\item The direct operational use of real values requires an explicit
representation, approximation, or canonicalization rule; these are
mathematically formalizable components of the protocol specification.
\item The central advantage concerns byte-level well-definedness without
additional serialization or rounding premises; no statement about cryptographic
security in the sense of secrecy, collision or second-preimage resistance is
implied thereby.
\end{itemize}

\section{Advantages of Fixed Algorithmic Presentations}

This section summarizes the insights gained above and puts them into systematic
order. The comparisons and the snaproot example show that the decisive
difference lies not in the numerical value but in the algorithmic
representability and canonicity of the representation.

\begin{itemize}
\item \textbf{Canonical representations:} values in the rational core
presentation $(\Q, \SQ)$ possess a uniquely fixed operational encoding.
\item \textbf{Reproducibility:} the same inputs lead to the same bytes in all
specification-conformant implementations.
\item \textbf{Protocol well-definedness:} the specification fixes inputs,
encodings, and processing steps in such a way that specification-conformant
implementations are interoperable without further implementation-dependent
auxiliary assumptions.
\item \textbf{Commitment well-definedness:} for commitments to concrete files,
the byte sequence itself is the object; for commitments to structured or
semantically interpreted objects, a canonical serialization is required insofar
as semantically equal initial representations are to produce the same commitment
input.
\item \textbf{Separation of operation and analysis:} rational core
presentations or other fixed operational representation systems for keys and
inputs; $\Aapp$ for error bounds, convergence, and analytic models.
\item \textbf{No additional decisions at run time:} once the normative encoding
has been fixed, no further implementation-dependent decisions about format,
rounding, or precision are required.
\end{itemize}

Important: these advantages concern formal protocol suitability and
reproducibility; they are no statement about the cryptographic security of the
toy constructions.

\paragraph{Natural objections.}
The following objections suggest themselves and are anticipated here.

\emph{Cryptography works with bit strings anyway.} Cryptographic protocols
operate with bit strings, integers, group and field elements --- so why the
detour via $\R$? $\Aex = \Q$ is a deliberately chosen rational core class and
not a maximal criterion for finite exact presentability: $\sqrt{2}$ and other
algebraic numbers, too, possess suitable exact finite representations, and
larger canonically representable systems (in particular for real algebraic
numbers~\cite{basupollackroy}) are possible. $\Aex$ is not intended as a
replacement for the objects mentioned, but as a conceptual model of the fact
that, and how, real values can be transformed into canonical operational objects
without additional approximation. The formal notion of the canonically
serializable object class (Definition~\ref{def:serializable}) transfers this
criterion to the finite objects occurring in practice --- files, hashes,
transaction IDs, payloads --- which are themselves not real numbers but share
the same property of canonical, finite representability.

\emph{Why not simply $\Z$?} A restriction to integers is accurate for many
cryptographic primitives but too restrictive for analysis purposes. The
advantage lies not in replacing $\Z$ but in the explicit separation between
operational objects (rational core presentation or canonically serializable
representations) and analytic quantities ($\Aapp$).

\emph{What does $\Aapp$ additionally provide?} Approximable numbers from $\Aapp$
make it possible to formulate convergence, error bounds, and numerical models
precisely without endangering operational well-definedness. $\Aapp$ thus covers
the analytic side that $\Aex$ deliberately leaves out.

\emph{Is this a security claim?} No. All advantages concern well-definedness,
reproducibility, and the formulation of byte-level correctness arguments; no
statement is made about the cryptographic security of the constructions used.

ASN.1/DER~\cite{x690} and the normatively fixed point, key, and signature
encodings in RFC~8032 (Ed25519)~\cite{rfc8032} are examples of tightly specified
operational encodings (RFC~8032 requires, for instance, $0 \le S < L$ and fixed
point and integer encodings; different verification equations can lead to
diverging sets of accepted signatures). snaproot~\cite{snaproot}, by contrast,
illustrates more generally the use of prescribed finite representations;
according to the present specification, the SR1 payload is not necessarily
canonical in the strong sense of Definition~\ref{def:canenc} (cf.\
Section~\ref{sec:snaproot}). All three can be understood, in the terminology of
this paper, as realizations of fixed operational representations, and they show
that the demanded separation between the operational and the analytic level is
already implicitly carried out in real systems.

\section{Limitations}

The scope of this paper is deliberately controlled. We state the most important
limitations openly:

\begin{itemize}
\item \textbf{Narrower notion of representation.} The notion of a representation
system used here is deliberately narrower than general Type-2 representations of
computable analysis; it suffices for the protocol questions considered but makes
no claim to their full generality.
\item \textbf{Classical foundations.} The main results are predominantly
reformulations of classical computability limits (undecidability of equality of
computable reals); the contribution lies in framing, transfer, and positioning,
not in a new undecidability theorem.
\item \textbf{Taxonomy not yet formalized.} The protocol-relative commitment
taxonomy sketched in Section~11 is a research program and not yet fully
formalized.
\item \textbf{No security analysis.} The toy schemes (XOR, toy RSA, toy hash)
serve exclusively to illustrate the representation question; no cryptographic
security analysis is carried out.
\item \textbf{Own case study.} The snaproot case study is the author's own work
and serves as illustration, not as independent external evidence.
\item \textbf{Specification relativity.} Canonicity is not an intrinsic property
of an abstract object but always relative to a normative syntax, normalization,
and encoding specification. Different specifications can fix different, each
internally canonical, byte encodings.
\item \textbf{No empirical evaluation.} The paper contains no empirical study of
competing implementations or real interoperability failures; the analysis is
conceptual and mathematical, not experimental.
\end{itemize}

\section{Outlook: Protocol-Relative Commitment Semantics}\label{sec:outlook}

As an outlook, not as a fully elaborated theory, we sketch a protocol-relative
extension of the representation perspective developed here; the following
considerations are not part of the formal main claims of this paper.

The distinction made in this paper between exact byte identity, canonical
serialization of structured objects, and canonical value presentation suggests a
more general protocol-relative perspective. Let $\Pi$ be a protocol and
$\mathrm{Obs}_{\Pi}$ the collection of the properties of a presentation $d$
observed by it. Two presentations could then count as
\emph{protocol-equivalent} if
\[
d_1 \equiv_{\Pi} d_2 \;\Leftrightarrow\;
\mathrm{Obs}_{\Pi}(d_1) = \mathrm{Obs}_{\Pi}(d_2).
\]

\paragraph{Example (measurement protocol).}
Let $L_{\mathrm{mes}}$ be the decidable language of normalized decimal strings
with exactly five decimal places. A protocol $\Pi$ first decodes
$d \in L_{\mathrm{mes}}$ as an exact rational fixed-point value
$\mathrm{parse}(d)$ and observes of it only its rounding to two decimal places
according to a normatively fixed rule --- here \emph{round half away from zero}
(in an exact half case, away from zero in absolute value). We have
\[
\mathrm{Obs}_{\Pi}(d) = \mathrm{round}_2(\mathrm{parse}(d)), \qquad
d_1 \equiv_{\Pi} d_2 \Leftrightarrow
\mathrm{round}_2(\mathrm{parse}(d_1)) = \mathrm{round}_2(\mathrm{parse}(d_2)).
\]
The associated computable normal form as a concrete code is
\[
N_{\Pi}(d) = \mathrm{format}_2(\mathrm{round}_2(\mathrm{parse}(d))),
\]
where $\mathrm{format}_2$ always produces exactly two decimal places (for
instance the string ``3.14''); the canonical target form per equivalence class
is thereby completely fixed as a finite code.

The exactly different inputs $3.14159$ and $3.14161$ are protocol-equivalent
under this $\Pi$ (both round to $3.14$), while $3.14000$ and $3.15000$ are not.
Since $\mathrm{parse}(d)$ yields, for $d \in L_{\mathrm{mes}}$, an exactly
represented rational fixed-point value, $\mathrm{round}_2$ is computable by
integer arithmetic; the protocol-relative equivalence is therefore decidable and
possesses a computable normal form --- unlike the universal exact value equality
of Theorem~\ref{thm:barrier}.

Decisive is the restriction to the decidable input language
$L_{\mathrm{mes}}$: on an unrestricted class of arbitrary approximation programs
for computable real numbers, a rounding normal form would not be totally
computable, since $\mathrm{round}_2$ is discontinuous at the rounding thresholds
($3.145$; $3.155$; \dots) and, from an arbitrary Cauchy name, it cannot be
decided in the limiting case on which side of the threshold the value lies. The
example thus demonstrates precisely the core thesis: computability arises
through the precise specification of input language and rounding rule, not
through the mere notion of rounding. It shows how a deliberately coarser,
protocol-defined equivalence remains well-defined on a fixed input class.

This equivalence need coincide neither with syntactic byte equality nor with
exact real value equality. In particular, it can be coarser than value equality
and thereby enable a canonical protocol-relative representation, even though a
universal exact value canonicalization is excluded by
Theorem~\ref{thm:barrier}. Precisely herein lies the most interesting starting
point of an independent protocol-relative theory.

A coarser protocol-relative equivalence weakens the requirements compared with
exact value equality, but by itself guarantees neither the decidability of the
equivalence nor the existence of a computable normal form; both must be
established separately for the respective protocol. Moreover, a
protocol-relative commitment binds, on the semantic level, only to the
equivalence class $[d]_{\Pi}$ distinguished by $\Pi$, not necessarily to a
single real value --- an observation central to the cryptographic binding
property.

From this, prospectively, a taxonomy emerges of \emph{representation
commitments} (commitment to the concrete description;
$N_{\mathrm{rep}}(d_1) = N_{\mathrm{rep}}(d_2) \Leftrightarrow d_1 = d_2$),
\emph{exact value commitments}
($N_{\mathrm{val}}(d_1) = N_{\mathrm{val}}(d_2) \Leftrightarrow
\sem{d_1} = \sem{d_2}$, i.e.\ exactly when the values are equal), and
\emph{protocol-relative commitments}
($N_{\Pi}(d_1) = N_{\Pi}(d_2) \Leftrightarrow d_1 \equiv_{\Pi} d_2$). Demanding
the respective equivalence in both directions ensures that the preprocessing
does not merge classes not provided for by the protocol. A complete theory would
additionally have to examine the compatibility of the equivalence with the
operations used by the protocol, the composition of canonically serializable
objects, and the demarcation from represented spaces, computable structures,
exact real arithmetic, and existing semantic commitment notions. This
elaboration is reserved for follow-up work.

The proposed taxonomy here concerns the semantics or canonicalization of the
input to be committed to. It makes no statement about determinism, hiding, or
binding of the commitment scheme used subsequently: classical schemes can be
randomized ($\mathrm{Commit}(m; r)$), so that even for identical input,
different random values produce different commitments. The taxonomy thus
describes how the semantic commitment input is determined from a presentation,
not the commitment scheme itself. Classical commitment
definitions~\cite{damgaard1999} presuppose an already fixed message space; the
perspective proposed here asks, upstream of that, which presentations count as
the same message or the same protocol-relevant equivalence class in the first
place.

\section{Discussion and Conclusion}

Operational cryptographic procedures process finite byte sequences. If
semantically equal objects are thereby to produce identical protocol bytes
independently of their initial representation (byte-input invariance), a
canonical representation or equivalent normalization must additionally be fixed.
The rational core presentation $(\Q, \SQ)$ satisfies this condition for rational
values; canonically serializable object classes transfer the same property to
practical protocol objects such as files, hashes, transaction IDs, and payloads.
Algorithmically approximable numbers $\Aapp$, by contrast, are appropriate for
analysis and estimation purposes, provided no byte-input-invariant processing is
demanded.

As a supplementary note: models such as the random oracle
model~\cite{bellarerogaway1993} presuppose already encoded bit strings as
inputs. The upstream mapping of semantic objects to these bit strings is
presupposed by the model but not itself fixed. The perspective presented here
lies upstream of it and treats exactly this representation layer --- an
operational view complementary to the idealization, which the model presupposes
but does not itself solve.

Beyond this, the correctness proof presented in this paper shows a more specific
advantage: a canonical presentation is not a prerequisite of mathematical
provability as such --- formal proofs can also be conducted via Cauchy names,
intervals, or other representations of real numbers. A canonical presentation is
a sufficient and practically particularly suitable way of realizing
representation-invariant processing; it is required in particular when
semantically equal objects are to be represented by identical canonical byte
inputs (byte-input invariance). It is, however, not necessary in general:
representation invariance can also be achieved by directly extensional
processing that is well-defined on equivalence classes --- for instance, by a
decoder accepting several notations and subsequently operating on the decoded
value without first computing a canonical normal form. A function on bytes is
uniquely specified even without canonicity; canonicity is needed so that
semantically equal objects receive the same canonical byte input independently
of their initial representation. In this sense, the rational core presentation,
or a canonically serializable representation, enables byte-level correctness
arguments that are reproducible independently of the implementation for an
identical byte sequence.

Beyond the toy examples considered here, the snaproot hash-anchoring protocol
shows that the more general separation between semantic objects and their
operational byte representations is also relevant for real hashing and
verification systems. In such contexts, algorithms operate on syntactic
representations, not on abstract mathematical objects. The use of fixed
canonical or protocol-conformant finite representations therefore makes it
possible to formulate protocols in such a way that their correctness, once the
representation specification has been fixed, is formally verifiable
independently of further, non-normalized implementation decisions.

The contribution of this paper thus lies neither in new cryptographic primitives
nor in a new undecidability theorem, but in a representation-theoretic framework
for cryptographic protocol objects: in the explicit separation between
approximable analysis objects ($\Aapp$) and canonically serializable operational
protocol objects (rational core presentation or protocol-conformant
representations) --- and in the observation that this separation forms a basis
for byte-input-invariant protocol processing and for byte-level correctness
arguments reproducible in all specification-conformant implementations. A
protocol that commits to concrete byte sequences can be well-defined even
without semantic canonicalization. Put pointedly, the contribution lies in the
protocol-theoretic reformulation of classical computability limits and their
connection with canonical serialization and commitment semantics.

For cryptographic protocols it is therefore decisive not only which mathematical
or semantic value is meant, but through which normative presentation it becomes
a reproducible operational object.

\section*{Statement on the Use of AI Tools}

In the preparation of this paper, AI language models were used as tools, in
particular Anthropic Claude Opus~4.8 and Claude Fable~5 as well as OpenAI
GPT-5.5 Thinking. They served for linguistic revision, critical review of the
argumentation, checking of the worked examples, and the translation of this
English version from the German original. The responsibility for the content,
all conceptual decisions, and the final version lie entirely with the author.


\appendix
\section{Detailed Calculations (Complete)}\label{app:hash}

\subsection{Symmetric Encryption (Complete Steps)}
\begin{center}
\begin{tabular}{l}
Key $k = 13/37 \in \Aex$\\
$\mathrm{pair}_{\mathrm{toy}}(13/37) = (p, q) = (13, 37)$\\
$s_i = (p + 31 \cdot q + i) \bmod 256$\\
$p + 31 \cdot q = 1160$\\
$S$ (decimal) = $[136, 137, 138, 139, 140]$\\
$S$ (hex) = \texttt{[88, 89, 8A, 8B, 8C]}\\
$P$ = ``HALLO'' (ASCII) = \texttt{[48, 41, 4C, 4C, 4F]}\\
$C = P \oplus S$:\\
\texttt{48$\oplus$88=C0, 41$\oplus$89=C8, 4C$\oplus$8A=C6, 4C$\oplus$8B=C7,
4F$\oplus$8C=C3}\\
$C$ (hex) = \texttt{[C0, C8, C6, C7, C3]}\\
Back-calculation: $P = C \oplus S$
\end{tabular}
\end{center}

\subsection{RSA (Modular Powers via Square-and-Multiply)}
Parameters: $p = 11$, $q = 17$, $n = 187$, $\varphi(n) = 160$, $e = 7$,
$d = 23$, $m = 42$.
\begin{center}
\begin{tabular}{l}
Encryption: $c = 42^7 \bmod 187$\\
$42^2 = 1764 \bmod 187 = 81$\\
$42^4 = 81^2 = 6561 \bmod 187 = 16$\\
$42^7 = 42^4 \cdot 42^2 \cdot 42 = 16 \cdot 81 \cdot 42 \bmod 187$\\
$16 \cdot 81 = 1296 \bmod 187 = 174$\\
$174 \cdot 42 = 7308 \bmod 187 = 15 \Rightarrow c = 15$\\[4pt]
Decryption: $m = 15^{23} \bmod 187$\\
$15^2 = 225 \bmod 187 = 38$\\
$15^4 = 38^2 = 1444 \bmod 187 = 135$\\
$15^8 = 135^2 = 18225 \bmod 187 = 86$\\
$15^{16} = 86^2 = 7396 \bmod 187 = 103$\\
$23 = 16 + 4 + 2 + 1$\\
$15^{23} = 103 \cdot 135 \cdot 38 \cdot 15 \bmod 187$\\
$103 \cdot 135 = 13905 \bmod 187 = 67$\\
$67 \cdot 38 = 2546 \bmod 187 = 115$\\
$115 \cdot 15 = 1725 \bmod 187 = 42 \Rightarrow m = 42$
\end{tabular}
\end{center}

\subsection{Toy Hash (Complete Sums)}
\[
H(x) = \sum_{i=0}^{L-1} (i+1) \cdot x_i \bmod 256
\]
\begin{center}
\begin{tabular}{l}
$k = 13/37$, $\mathrm{bytes}_{\mathrm{toy}}(13/37) = [13, 37]$:\\
$H = (1 \cdot 13 + 2 \cdot 37) \bmod 256 = 87$\\[4pt]
$k = \pi$, ``3.14'' $\rightarrow [51, 46, 49, 52]$:\\
$H = 51 + 92 + 147 + 208 = 498 \bmod 256 = 242$\\[4pt]
$k = \pi$, ``3.1416'' $\rightarrow [51, 46, 49, 52, 49, 54]$:\\
$H = 1067 \bmod 256 = 43$
\end{tabular}
\end{center}
Full SHA-256 values for Section~\ref{sec:sha} (hexadecimal, 256 bits, each one
contiguous value):
\begin{center}
\footnotesize
\begin{tabular}{l}
SHA-256(``3.14'') =\\
\texttt{2efff1261c25d94dd6698ea1047f5c0a7107ca98b0a6c2427ee6614143500215}\\[3pt]
SHA-256(``3.1416'') =\\
\texttt{23741492d396d4d556f0982fbcc97870d7b0e61859d7d7619d8418f875a42602}\\[3pt]
SHA-256($[13, 37]$) =\\
\texttt{51698d9afde54286cf8875ec5c4163a3b435e5f07e599ce20349fb342f94f015}
\end{tabular}
\end{center}

\subsection{Note on the Term ``Toy''}
The term ``toy'' denotes, in this context, deliberately simplified, insecure
example constructions that serve exclusively to illustrate structural and
conceptual properties. In particular, no statement about cryptographic security
is made. This is a common convention in cryptography and theoretical computer
science.

\end{document}